\documentclass[10pt,jounral]{IEEEtran}

\usepackage{graphics,
           psfrag,
           epsfig,
           amsmath,
           amsthm,
           cite,
           amssymb,
           url,
           dsfont,
           algorithm,
           algorithmic,
           balance,
           enumerate,
           color,
           setspace,
           multirow,
           caption,
           subcaption,
           comment
}

\newtheorem{definition}{Definition}

\newtheorem{theorem}{Theorem}
\newtheorem{corollary}{Corollary}

\newcommand{\sref}[1]{Section~\ref{#1}}

\newcommand{\fref}[1]{Figure~\ref{#1}}

\newcommand{\cref}[1]{Constraint~\ref{#1}}
\newcommand{\thref}[1]{Theorem~\ref{#1}}

\newcommand{\ignore}[1]{}

\IEEEoverridecommandlockouts
\begin{document}
%\doublespacing
%\includecomment{singlecol}\excludecomment{doublecol} % Single Column
\includecomment{doublecol}\excludecomment{singlecol} % Double Column

\title{Partially Blind Instantly Decodable Network Codes for Lossy Feedback Environment}
\author{Sameh~Sorour,~\IEEEmembership{Member,~IEEE,} Ahmed~Douik,~\IEEEmembership{Student Member,~IEEE,}
        Shahrokh~Valaee,~\IEEEmembership{Senior Member,~IEEE}, Tareq~Y.~Al-Naffouri,~\IEEEmembership{Senior Member,~IEEE} Mohamed-Slim~Alouini,~\IEEEmembership{Fellow,~IEEE}
\thanks{Sameh Sorour and Mohamed-Slim Alouini are with the Computer, Electrical and Mathematical Sciences and Engineering (CEMSE) Division, King Abdullah University of Science and Technology (KAUST), Thuwal, Makkah Province, Saudi Arabia, email:$\{\mbox{sameh.sorour, slim.alouini}\}$@kaust.edu.sa}
\thanks{Ahmed Douik is with both Tunisia Polytechnic School (TPS), Tunisia, and the CEMSE Division at King Abdullah University of Science and Technology (KAUST), Thuwal, Makkah Province, Saudi Arabia, email: ahmed.douik@kaust.edu.sa}
\thanks{Shahrokh Valaee is with the Edward S. Rogers Sr. Department of Electrical and Computer Engineering,
    University of Toronto, Canada, e-mail: valaee@comm.utoronto.ca.}
    \thanks{Tareq Y. Alnaffouri is with both the CEMSE Division at King Abdullah University of Science and Technology (KAUST), Thuwal, Makkah Province, Saudi Arabia, and the Electrial Engineering Department at King Fahd University of Petroleum and Minerals (KFUPM), Dhahran, Eastern Province, Saudi Arabia, e-mail: tareq.alnaffouri@kaust.edu.sa.}
    \thanks{This work is an extension to the first and third authors' paper \cite{PIMRC11}, published in the IEEE International Symposium on Personal, Indoor and Mobile Radio Communications (PIMRC), September 2011.}
     }

\maketitle

\IEEEoverridecommandlockouts

\begin{abstract}
In this paper, we study the multicast completion and decoding delay minimization problems of instantly decodable network coding (IDNC) in the case of lossy feedback. In such environments, the sender falls into uncertainties about packet reception at the different receivers, which forces it to perform partially blind selections of packet combinations in subsequent transmissions. To determine efficient partially blind policies that handle the completion and decoding delays of IDNC in such environment, we first extend the perfect feedback formulation in \cite{GC10,TON10-CD} to the lossy feedback environment, by incorporating the uncertainties resulting from unheard feedback events in these formulations. For the completion delay problem, we use this formulation to identify the maximum likelihood state of the network in events of unheard feedback, and employ it to design a partially blind graph update extension to the multicast IDNC algorithm in \cite{TON10-CD}. For the decoding delay problem, we derive an expression for the expected decoding delay increment for any arbitrary transmission. This expression is then used to derive the optimal policy to reduce the decoding delay in such lossy feedback environment. Results show that our proposed solution both outperforms other approaches and achieves a tolerable degradation even at relatively high feedback loss rates.
\end{abstract}
\ignore{}

\section{Introduction}
\IEEEPARstart{N}{etwork} coding (NC) has shown great abilities to substantially improve transmission efficiency, throughput and delay over broadcast erasure channels \cite{4895447,4476183,Drinea2009,4397057,5152148}. In \cite{ICC10,Sadeghi2010,5425315,GC10,Li2011,6030131}, an important sub-class of opportunistic NC was studied, namely the \emph{instantly decodable network coding (IDNC)}, in which received packets are allowed to be decoded only at their reception instant and cannot be stored for future decoding. IDNC was considered in these works due to its practicality and numerous desirable properties, such as instant packet recovery, simple XOR-based packet encoding and decoding, and no additional buffer to store un-decoded packets. According to its definition, the sender must select a network coded packet combination in each transmission, such that a selected subset of the receivers (or all of them if possible) can decode a new source packet once they receive this coded packet. The selection of the appropriate packet combinations, which are instantly decodable at specific sets or all the receivers, is done through what is known as the \emph{IDNC graph} \cite{ISIT09,ICC10,GC10,TON10-CD,6030131}.

The problem of minimizing the completion delay for IDNC was considered in \cite{ICC10,TON10-CD} for wireless networks with erasure channels on the forward links from the sender to the receivers. This problem was formulated as a stochastic shortest path (SSP) problem, which turned out to be very complex to solve in real-time. Nonetheless, the analysis of the properties and structure of this SSP was employed to design simple maximum weight clique search algorithms minimizing completion delay in IDNC. The designed algorithms were shown to almost achieve the optimal completion delay in wireless multicast and broadcast scenarios. In \cite{GC10}, the \emph{decoding delay} minimization problem in IDNC was considered for the same environment. In this work, an expression for the \emph{decoding delay} increments for any arbitrary IDNC transmission was derived and used to minimize the IDNC decoding delay, again using simple maximum weight clique search algorithms.

Nonetheless, the proposed algorithms in \cite{ICC10,TON10-CD,GC10}, and most other opportunistic network coding works, assume perfect feedback (PF) (i.e feedback from all the receivers always arrives at the sender and is not subject to loss). This assumption is not always practical in wireless networks, which suffer from similar wireless channel impairments on both the forward and reverse links. Consequently, even if a high level of protection for feedback packets can be employed in several network settings, such as cellular, WiFi and WiMAX systems, unavoidable occasions of deep fading over wireless channels can still expose them to loss events. Moreover, other network settings cannot guarantee the arrival of each feedback packet at the sender due to transmission power limitations and possible interference with other feedback.

In such lossy feedback scenarios, the sender may not receive feedback from a subset of the targeted receivers after any given transmission, due to erasure occasions on their reverse links. In this case, the reception status for this packet at these receivers becomes uncertain. For each of these receivers, the sender cannot determine whether its sent packet was lost on the forward link of this receiver or the feedback was lost on its reverse link. Despite this uncertainty, the sender is still required to perform subsequent IDNC transmissions, and thus must blindly estimate the status of all such receivers. In subsequent transmissions, the sender may receive feedback packets from some of these uncertain receivers but may not hear the feedback from others. Consequently, the sender may need to perform partially blind IDNC decisions in many transmissions until a completion feedback is received from all the receivers.

In this paper, we address the following question: \emph{How to perform partially blind selections of IDNC packet combinations to minimize the IDNC completion and decoding delays in lossy feedback (LF) scenarios?} To answer this question, we need to extend our formulations, derivations and proposed algorithms in \cite{ICC10,TON10-CD,GC10} by incorporating in them the uncertainties resulting from unheard feedback events. For the completion delay problem, we extend our SSP formulation for the perfect feedback scenario in \cite{TON10-CD} to a partially observable stochastic shortest path (POSSP) problem \cite{Patek99onpartially} in the lossy feedback environment. From this extended formulation, we derive the maximum likelihood estimate of the network state in case of unheard feedback events from one or multiple receivers. We then propose a partially blind IDNC algorithm that makes coding decisions\ignore{ as in \cite{TON10-CD},} after performing partially blind updates on the IDNC graph to follow the maximum likelihood state of the network. For the decoding delay problem, we derive the expected decoding delay increments for any arbitrary transmission in the presence of reception uncertainties. This derived expression is then used to define the new policy for decoding delay minimization in lossy feedback environment. Using extensive simulations, we compare the proposed algorithm with the several other approaches available in the literature.

The rest of the paper is organized as follows. In \sref{sec:model}, we introduce the system model
and parameters. The IDNC graph is illustrated in \sref{sec:IDNC-graph}. In \sref{sec:CD-extension}, we present the POSSP extension of the PF SSP formulation in the lossy feedback environment, derive the maximum belief state (i.e. maximum likelihood state) of the network state, and propose a partially blind IDNC algorithm to solve the completion delay problem in lossy feedback environments. In \sref{sec:DD-extension}, the expected decoding delay increments of any transmission is derived for the lossy feedback environment and is used to design a partially blind algorithm to minimize the decoding delay. The performance of the proposed algorithms is compared with several algorithms in the literature in \sref{sec:simulations}. Finally, \sref{sec:conclusion} concludes the paper.

\section{System Model and Parameters} \label{sec:model}
The system model consists of a wireless sender having a frame (denoted by $\mathcal{N}$) of $N$ source packets. Each receiver in the set (denoted by $\mathcal{M}$) of $M$ receivers is interested in receiving a subset of $\mathcal{N}$. We will refer to the requested and undesired packets by any receiver as its ``primary'' and ``secondary'' packets, respectively. The sender initially broadcasts the $N$ packets uncoded over erasure channels to the $M$ receivers. Each receiver listens to all the transmitted packets (even the ones that it does not want) and feeds back to the sender a positive acknowledgement (ACK) for each received (non-erased) packet and a negative acknowledgement (NACK) for each erased packet. These feedback packets are also subject to erasure on the reverse channels. We will refer to the packets, for which the sender did not receive a feedback from a given receiver, by the ``uncertain'' packets of that receiver. At the end of this initial transmission phase, the sender can attribute four feedback sets of packets to each receiver $i$:
\begin{itemize}
\item The \emph{Has} set (denoted by $\mathcal{H}_i$) is defined as the set of primary and secondary packets correctly received by receiver $i$ and its feedback was successfully received by the sender.\ignore{ This set includes both desired and undesired packets by this receiver.}
\item The \emph{Lacks} set (denoted by $\mathcal{L}_i$) is defined as the set of primary and secondary packets that were either acknowledged by a NACK to the sender by receiver $i$ or are uncertain. In other words, $\mathcal{L}_i = \mathcal{N} \setminus \mathcal{H}_i$.
\item The \emph{Wants} set (denoted by $\mathcal{W}_i$) is defined as the subset of primary packets in $\mathcal{L}_i$ (i.e. packets in $\mathcal{L}_i$ that receiver $i$ wants to receive).
\item The Uncertain set (denoted by $\mathcal{U}_i)$ is defined as the subset of uncertain packets in $\mathcal{L}_i$.
\end{itemize}
The sender stores this information in a \emph{state feedback matrix (SFM)} $\mathbf{F}_\mathds{U} = [f_{ij}]~\forall~i\in\mathcal{M},j\in\mathcal{N}$, such that:
\begin{equation}
f_{ij} =
\begin{cases}
0 &\qquad  j \in \mathcal{H}_i \\
- 1 &\qquad  j \in \mathcal{L}_i\setminus \left(\mathcal{W}_i \cup \mathcal{U}_i\right)\\
1 &\qquad j \in \mathcal{W}_i\setminus\mathcal{U}_i\\
\mathbf{x} & \qquad j\in\mathcal{U}_i\;.
\end{cases}
\end{equation}
In the recovery phase, the sender exploits these feedback sets to transmit network coded packets, which consist of a binary XOR of a subset or all the source packets in $\mathcal{N}$. Each of these transmitted coded packets can be one of the following three options for each receiver $i$:
\begin{itemize}
\item \emph{Non-Innovative}: A packet is non-innovative for receiver $i$ if it does not contain any source packet that is both wanted by $i$ and was not received by it in any previous transmission.
\item \emph{Instantly Decodable}: A packet is instantly decodable for receiver $i$ if it contains \emph{only one source packet} that was not received by it in any previous transmission.
\item \emph{Non-Instantly Decodable}: A packet is non-instantly decodable for receiver $i$ if it contains {two or more source packets} that were not previously received by it in any previous transmission..
\end{itemize}
The receivers for which the packet sent by the sender is instantly decodable are called \emph{targeted receivers}. The receivers that receive non-innovative and non-instantly decodable packets discard them and do not send any acknowledgements. On the other hand, the receivers that receive instantly decodable packets feed back an ACK packet to the sender with all their received packets. This process is repeated until all receivers obtain all their requested packets and the sender receives a completion feedback acknowledgement from all the receivers (i.e. a feedback from every receiver showing that it received all its requested packets). To be fair in comparison with the original perfect feedback formulation and algorithms in \cite{ICC10,TON10-CD}, in terms of feedback frequency, we will assume that a receiver does not send any feedback unless it is targeted by a packet.

Let $p_i$ and $\overline{p}_i = 1-p_i$ be the data packet erasure and success probabilities, respectively, observed by receiver $i$ on the forward link within a frame of packets. Also, let $q_i$ be the feedback erasure probability observed by receiver $i$ on the reverse link. It is fair to assume that $p_i > q_i$ due to the larger size of data packets compared to feedback packets, the stronger protection usually employed for control packets, and the stronger interference levels and fading levels experienced by the receivers (especially those at cell edges) compared to the sender. Also, let $\mu_i$ be the demand ratio of receiver $i$, defined as the ratio of its primary packets in the frame to the total frame size $N$. Finally, define $\mu = \frac{1}{M}\sum_{i=1}^M \mu_i$ as the average demand ratio of all receivers.

Finally, we define the completion and decoding delays as follows:
\begin{definition}[Completion Delay]
The completion delay of a frame is the number of recovery transmissions required until all receivers obtain all their requested packets.
\end{definition}
\begin{definition}[Decoding Delay]
For any transmission in recovery phase, a receiver $i$, with non-empty Wants set, experiences one unit increase of \emph{decoding delay} if it successfully receives a packet that does not allow it to decode a non-previously received source packet in its Wants set. In other words, a received packet at receiver $i$ does not increase its decoding delay if and only if both following conditions are satisfied:
\begin{itemize}
\item This received packet is instantly decodable for that receiver.
\item The decodable source packet from this received packet is in its Wants set.
\end{itemize}
\end{definition}
Note that the definition of decoding delay does not count channel inflicted delays due to erasures (i.e. an erased packet at a receiver does not add to its decoding delay), but rather counts sender inflicted delays when its coding algorithm is not able to provide instantly-decodable packets to the different receivers.

\section{The IDNC Graph} \label{sec:IDNC-graph}
The IDNC graph provides a framework to determine all possible combinations of source packets that are instantly decodable for any subset or all the receivers. This graph $\mathcal{G}$ is constructed by first generating a vertex $v_{ij}$ in $\mathcal{G}$ for each packet $j \in \mathcal{L}_i$, and for all receivers (i.e. $\forall~i\in\mathcal{M}$. Two vertices $v_{ij}$ and $v_{kl}$ in $\mathcal{G}$ are adjacent if one of the following conditions is
true:
\begin{itemize}
\item C1: $j = l$, i.e. the two vertices are induced by the loss of the same packet $j$ by two different receivers $i$ and $k$.
\item C2: $j\in \mathcal{H}_k$ and $l \in \mathcal{H}_i$, i.e. the requested packet of each vertex is in the Has set of the receiver that induced the other vertex.
\end{itemize}
Consequently, each edge between two vertices $v_{ij}$ and $v_{kl}$ in the graph means that the source packets $j$ and $l$ can be simultaneously received/decoded at receivers $i$ and $k$, respectively, by sending either packet $j$ if $j=l$ or the coded packet $j\oplus l$ otherwise. This property extends from two adjacent vertices to every clique in the graph. A clique in a graph is a subset of this graph whose vertices are all adjacent to one another. Thus, each clique in  $\mathcal{G}$ defines a packet combination that can instantly serve all the receivers inducing this clique's vertices. Since we are concerned with not missing an opportunity of serving any possible receivers in any transmission, we only consider maximal cliques (i.e. cliques that are not a subset of any larger cliques).

According to receivers' packet requests, we can classify the vertices of this graph into two layers:
\begin{itemize}
\item Primary graph $\mathcal{G}_\rho$: It includes all the vertices of each receiver $i$, $\forall~i\in\mathcal{M}$, which are in its Wants sets.
\item Secondary graph $\mathcal{G}_\sigma$: It includes all the vertices of each receiver $i$, $\forall~i\in\mathcal{M}$, which are in its Lacks set but not in its Wants set.
\end{itemize}
\fref{fig:IDNC-graph} depicts an example of SFM and its corresponding two-layered IDNC graph.
\begin{figure}[t]
\centering
  \includegraphics[width=0.55\linewidth]{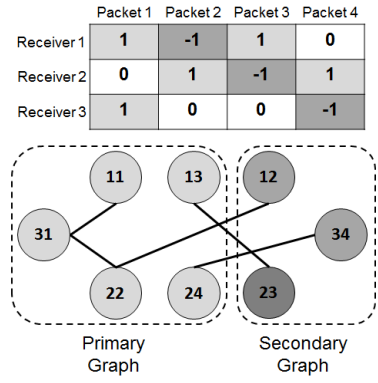}\\
  \caption{Example of a state feedback matrix and its corresponding IDNC graph. Each vertex denoted by the two digits $ij$ in the graph represents a 1 or -1 entry in the $i$-th row and $j$-th column of feedback matrix. The light and dark shaded boxes and vertices represent the requested (1 entries) and undesired packets (-1 entries), respectively.}\label{fig:IDNC-graph}
\end{figure}
Since the primary graph is the one that needs to be depleted by the sender (as it is the one that represents the receiver requests), the sender needs to mainly serve the packet combination represented by the best maximal clique $\kappa_\rho$ in the primary graph for any given transmission, without considering the secondary graph. Nonetheless, it can also benefit from the secondary graph to deliver undesired packets to non-targeted receivers by $\kappa_\rho$ (i.e. receivers not having vertices in $\kappa_\rho$) in the same transmission, without violating the instant decodability condition of the primary maximal clique $\kappa_\rho$. Serving these packets can increase the coding opportunities of these receivers in subsequent transmissions. This step can be done by finding the best secondary maximal clique $\kappa_\sigma$ from the connected secondary subgraph to the primary maximal clique $\kappa_\rho$.

Thus, the sender generates a packet combination for any given transmission by building the IDNC graph representing the SFM just before this transmission, and XORing all the packets identified by the vertices of the selected clique $\kappa = \kappa_\rho \cup \kappa_\sigma$ from this IDNC graph. In the rest of the paper, by ``the transmission of $\kappa$'' we mean the transmission of the XOR of the packets identified from clique $\kappa$ in the IDNC graph.

\section{Completion Delay Problem} \label{sec:CD-extension}
In this section, we study the effect of having feedback loss probabilities on the formulation and solution of the completion delay problem in IDNC. We will start by presenting the formulation and solution for the perfect feedback environment and then extend them to the lossy feedback environment.

\subsection{SSP Formulation for Perfect Feedback Environment}\label{sec:perfect-formulation}
The minimum completion delay problem in perfect feedback scenario is formulated in \cite{ICC10,TON10-CD} as an SSP. It consists of a state space $\mathcal{S}$ representing all SFM possibilities from the start of the recovery phase until completion. The action space of each state consists of all packet combinations identified by its SFM's IDNC graph. Transition probabilities between states reflect the taken action (in our case the chosen clique for transmission) and the different erasure probabilities of target receivers. Finally, the cost of each action is one transmission.

In \cite{TON10-CD}, the properties of this intractable SSP were analyzed and it has been shown that the best strategy to reduce the completion delay is to give more priority to the receivers with the largest Wants sets and erasure probabilities. This can be done by assigning weights $\psi_i = \left(\frac{|\mathcal{W}_i|}{\overline{p}_i}\right)^m$ ($m$ is a biasing factor) to the vertices of each receiver $i$ in the IDNC graph, then running maximum weight clique algorithms on the primary graph and then on the secondary graph.

\subsection{Lossy Feedback Extension: Belief State}\label{sec:lossy-formulation}
The difference between perfect and lossy environments is the uncertainties introduced at the sender due to unheard feedback occurrences. In perfect feedback environment, unheard feedback events at the sender from a targeted receiver make the sender certain that the sent packet was lost on the receiver's forward link. However, the lossy feedback environment adds the other possibility of packet reception on the receiver's forward link and the loss of the feedback on its reverse link. Each of these possibilities happens independently for each of the targeted receivers with unheard feedback. In this case, the sender cannot determine the exact SFM state of the network, and thus cannot accurately select the subsequent IDNC packet. This notion of uncertainty is illustrated in \fref{fig:LF-example}, depicting an example of the system in \fref{fig:IDNC-graph} after performing action $1\oplus 4$.

\begin{figure}[t]
\centering
  \includegraphics[width=0.9\linewidth]{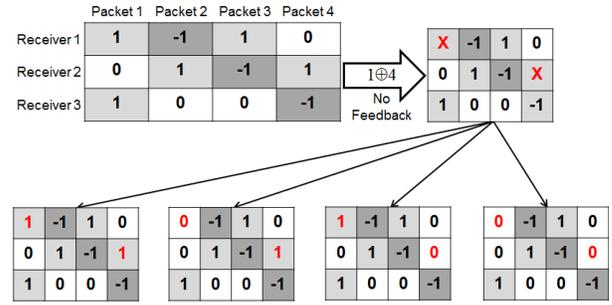}
    \caption{Belief state after taking action $1\oplus 4$ for the example of \fref{fig:IDNC-graph} and receiving no feedback from receivers 1 and 2. This results in the $\mathbf{x}$ entries in the top right matrix. Consequently, there exists four different possibilities for the actual state of the matrix shown in the four bottom matrices. Thus, there exists a non-zero probability that the actual state of the matrix is one of these four bottom matrices. The belief state consists of these probabilities.}\label{fig:LF-example}
\end{figure}

This resulting uncertainty converts the SSP formulation of the perfect feedback environment into a partially observable SSP (POSSP) problem \cite{Patek99onpartially}. This POSSP formulation of the lossy feedback scenario is an extension to the SSP of the perfect feedback scenario, by adding to it the POSSP's belief state. In a POSSP, the belief state $\mathbf{b}\left(\mathbf{F}_\mathds{U}\right) = \left[b\left(s\big|\mathbf{F}_\mathds{U}\right)\right]~\forall~s\in\mathcal{S}$ of an SFM $\mathbf{F}_\mathds{U}$ is defined as the probability distribution vector over the state space $\mathcal{S}$, where each element $b\left(s\big|\mathbf{F}_\mathds{U}\right)$ denotes the probability that the system is in state $s$. This POSSP belief state reflects all possible realizations of the entries $\mathbf{x}$ in the SFM, as shown in \fref{fig:LF-example}. Since each of these $\mathbf{x}$ entries arises at element $f_{ij}$ of the SFM when the sender targets receiver $i$ with packet $j$ in a coded recovery transmission and does not hear a feedback from it, it could be equal to $\{1,-1\}$ if $i$ did not receive $j$ or $0$ otherwise. Consequently, the states with non-zero values in $\mathbf{b}\left(\mathbf{F}_\mathds{U}\right)$ are those representing all possible combinations of replacing each $\mathbf{x}$ entry of $\mathbf{F}_\mathds{U}$ by 0 (assuming it is received) or $\{1,-1\}$ (assuming it is not received). The number of such states is equal to $2^{\left(\sum_{i=1}^M \left|\mathcal{U}_i\right|\right)}$.

\subsection{Maximum Likelihood State} \label{sec:belief}
Since the only difference between the perfect and lossy formulations is the uncertainty introduced by unheard feedback events, the same solution proposed in \sref{sec:perfect-formulation} can be adapted to solve the POSSP problem, if we can find a good estimate of the SFM (and thus the IDNC graph) in these events. In this stochastic partially observable domain, the best state estimate of the SFM and the IDNC graph is the one representing the maximum likelihood (ML) state of the network, i.e. the state that has the highest probability in the belief state, defined as follows:
\begin{equation}\label{eq:ML-rule}
s^* = \arg\max_{s\in\mathcal{S}}\left\{b\left(s\big|\mathbf{F}_\mathds{U}\right)\right\}\;.
\end{equation}

To derive the expression for this ML state, we define $\lambda_{i}(s) \subset \mathcal{U}_i$ and $\overline{\lambda}_{i}(s) = \mathcal{U}_i\setminus \lambda_{i}(s)$ as the sets of $\mathbf{x}$ entries in the $i$-th row of the uncertain SFM $\mathbf{F}_\mathds{U}$ that need to be replaced by $\{1,-1\}$ and zeros, respectively, in order to reach a given realization $\mathbf{F}(s)$ of state $s$. Also, let $\theta_{ij}$ be the number of attempts (i.e. times) the sender targeted receiver $i$ with packet $j$ after the last time the sender heard a feedback from this receiver. Given these definitions, the following theorem introduces the expression for this ML state.
\begin{theorem}\label{th:ML}
The ML state of the network can be determined by the sender as:
\begin{equation}\label{eq:ML}
s^* =  \arg\max_{s\in\mathcal{S}}\left\{\prod_{i\in\mathcal{M}}\left( \prod_{j\in\lambda_i(s)} \mathcal{P}^L_{ij} \cdot \prod_{j\in\overline{\lambda}_i(s)} \mathcal{P}^R_{ij}\right)\right\}\;,
\end{equation}
where
\begin{equation}\label{eq:ML-individual}
\mathcal{P}^L_{ij} = \left(\frac{p_i}{p_i+\overline{p}_iq_i}\right)^{\theta_{ij}}\qquad\mbox{and}\qquad
\mathcal{P}^R_{ij} = 1-\mathcal{P}^L_{ij}\;.
\end{equation}
\end{theorem}

\begin{proof}
As can be inferred from \eqref{eq:ML-rule}, we need to derive an expression for the different entries $b\left(s\big|\mathbf{F}_\mathds{U}\right)$ (i.e. the probability distribution over the POSSP states) of the belief state in order to find the ML state of the network. Let $t_i$ be the last timeslot at which the sender heard a feedback from receiver $i$. Also, let $\mathcal{P}^R_{ij}$ $\left(\mathcal{P}^L_{ij}\right)$ be the probability that an uncertain packet $j$ of receiver $i$ has been actually received (not received), given the event of unheard feedback from receiver $i$ since time $t_i$. Finally, define $s_a$ as the actual hidden state of the network. Thus, we can express $b\left(s\big|\mathbf{F}_\mathds{U}\right)$ as:
\begin{singlecol}
\begin{equation} \label{eq:belief-computation}
b\left(s\big|\mathbf{F}_\mathds{U}\right) = \mathds{P}\left\{s_a = s \mid \mathbf{F}_\mathds{U}\right\} = \prod_{i\in\mathcal{M}} \left(\prod_{j\in\lambda_i(s)} \mathcal{P}^L_{ij} \cdot \prod_{j\in\overline{\lambda}_i(s)} \mathcal{P}^R_{ij} \right)\;.
\end{equation}
\end{singlecol}
\begin{doublecol}
\begin{align} \label{eq:belief-computation}
b\left(s\big|\mathbf{F}_\mathds{U}\right) &= \mathds{P}\left\{s_a = s \mid \mathbf{F}_\mathds{U}\right\} \nonumber \\
& = \prod_{i\in\mathcal{M}} \left(\prod_{j\in\lambda_i(s)} \mathcal{P}^L_{ij} \cdot \prod_{j\in\overline{\lambda}_i(s)} \mathcal{P}^R_{ij} \right)\;.
\end{align}
\end{doublecol}

From \eqref{eq:belief-computation}, we can define $s^*$ as the state satisfying the ML estimation, such that:
\begin{equation}\label{eq:ML-state}
s^* = \ignore{\arg\max_{s\in\mathcal{S}}\left\{b(s\big|\mathbf{F}_\mathds{U}\right)\right\} =} \arg\max_{s\in\mathcal{S}}\left\{\prod_{i\in\mathcal{M}} \left(\prod_{j\in\lambda_i(s)} \mathcal{P}^L_{ij} \cdot \prod_{j\in\overline{\lambda}_i(s)} \mathcal{P}^R_{ij} \right)\right\}\;.
\end{equation}
To find the state maximizing the right hand side, we need to derive expressions for $\mathcal{P}^L_{ij}$ and $\mathcal{P}^R_{ij}$. An unheard feedback occurrence from a targeted receiver $i$ in any give transmission could mean one of the two following events:
\begin{enumerate}
\item Receiver $i$ did not receive the packet and thus did not issue a feedback.
\item Receiver $i$ received the packet and issued a feedback packet, which did not arrive at the sender.
\end{enumerate}
These events can occur with probabilities $p_i$ and $\overline{p}_iq_i$, respectively. Note that unheard feedback events are independent from each other even for the different attempts of the same packet to the same receiver. They are also independent from the transmitted packet to receiver $i$ and independent from one receiver to another. Given the definition of $\theta_{ij}$ as the number of times the sender targeted receiver $i$ with packet $j$ after time $t_i$, the probability that an $\mathbf{x}$ entry at position $f_{ij}$ of $\mathbf{F}_\mathds{U}$ is in fact $\{1,-1\}$ (i.e. packet $j$ is actually not received at receiver $i$) is equal to the probability that this packet was lost on the forward channel for each of the $\theta_{ij}$ attempts. This happens with probability $\left(p_i\right)^{\theta_{ij}}$. On the other hand, this uncertain entry $f_{ij}$ will be 0 (i.e. packet $j$ is actually received by receiver $i$) if this receiver has successfully decoded it in any one, several or all attempts out of the $\theta_{ij}$ attempts. This occurs with probability $\sum_{k=0}^{\theta_{ij}-1} \binom{\theta_{ij}}{k} \left(p_i\right)^k\left[\overline{p}_iq_i\right]^{\theta_{ij}-k}$. Consequently, given $\theta_{ij}$ attempts of packet $j$ to receiver $i$ with no feedback from it on any of them, $\mathcal{P}^L_{ij}$ and $\mathcal{P}^R_{ij}$ can be expressed as:
\begin{align}
\mathcal{P}^L_{ij} &= \frac{\left(p_i\right)^{\theta_{ij}}}{\left(p_i\right)^{\theta_{ij}} + \sum_{k=1}^{\theta_{ij}-1} \binom{\theta_{ij}}{k} \left(p_i\right)^k\left[\overline{p}_iq_i\right]^{\theta_{ij}-k}} \nonumber \\
& = \left(\frac{p_i}{p_i + \overline{p}_iq_i}\right)^{\theta_{ij}}\;, \label{eq:ML-Loss} \\
&\quad \nonumber\\
\mathcal{P}^R_{ij} &= \frac{\sum_{k=1}^{\theta_{ij}-1} \binom{\theta_{ij}}{k} \left(p_i\right)^k\left[\overline{p}_iq_i\right]^{\theta_{ij}-k}}{\left(p_i + \overline{p}_iq_i\right)^{\theta_{ij}}} \nonumber\\
&= 1 - \mathcal{P}^L_{ij} \;. \label{eq:ML-Receive}
\end{align}
\end{proof}

In case of reciprocal channels, in which $p_i = q_i = p_i$, the ML expression can be simplified as defined by the following corollary.
\begin{corollary}\label{th:ML-reciprocal}
For reciprocal channels, the ML state of the network can be determined by the sender using \eqref{eq:ML}, where:
\begin{equation}\label{eq:ML-reciprocal-individual}
\mathcal{P}^L_{ij} = \left(\frac{1}{2-p_i}\right)^{\theta_{ij}} \qquad \mbox{and} \qquad
\mathcal{P}^R_{ij} = 1-\mathds{P}^L_{ij}\;.
\end{equation}
\end{corollary}
\begin{proof}
The proof can be easily obtained by a simple substitution of $q_i = p_i$ in the analysis to derive \eqref{eq:ML-Loss} and \eqref{eq:ML-Receive}.
\end{proof}

\subsection{Complexity and Simple ML Rules}
The question now is whether the sender needs to compute $2^{\left(\sum_{i=1}^M \left|\mathcal{U}_i\right|\right)}$ probabilities before each transmission in order to estimate the ML state of the network. The expressions in \eqref{eq:ML}, \eqref{eq:ML-individual} and \eqref{eq:ML-reciprocal-individual} in Theorem \ref{th:ML} and Corollary \ref{th:ML-reciprocal} show that the problem of estimating the ML state does not really need to go through all these computations after each transmission. These equations show that the independence of unheard feedback events for each receiver and between different receivers, as well as the independence of loss events on the forward and reverse channels, allow progressive computation of the ML state estimate of the network after each transmission as follows. When the sender experiences an unheard feedback event from receiver $i$, after a given transmission targeting this receiver with packet $j$, it can easily compute the individual ML state of this packet by computing $\mathcal{P}^L_{ij}$. If this probability is larger than 0.5, then it can set the corresponding entry $f_{ij}$ to 1 or -1 (according to whether it is wanted or not by receiver $i$, respectively). Otherwise, it sets $f_{ij}$ to zero. This update can be also simply done directly to the IDNC graph by keeping or eliminating the vertice $v_{ij}$ from the graph if $\mathcal{P}^L_{ij}$ is larger or smaller that 0.5, respectively.

Once an estimate is made for the entry $f_{ij} = \mathbf{x}$ and $v_{ij}$ after this transmission, their state need not be computed again for any subsequent transmission that does not involve them. In other words, if the sender does not target receiver $i$ with packet $j$ after this ML state estimate has been made, and if it did not receive any other feedback from this receiver, then it does not need to make any further estimates about this $f_{ij} = \mathbf{x}$ entry. This allows the progressive construction of the ML state and avoids unnecessary computations for the entries and receivers with unchanged status.

In addition to the complexity of the above progressive ML state estimation after each transmission, the sender needs to store the variables $\theta_{ij}$, which require $O(MN)$ entries in its memory. Finally note that, once a valid feedback is received from receiver $i$, all its entries $f_{ij}$ $\forall~j\in\mathcal{U}_i$ can be updated in the SFM, its vertices can be eliminated or re-inserted in the graph accordingly, and the values of $\theta_{ij}$ of this receiver are all reset.

The following corollaries can be used to further simplify the computation of the ML estimates for the packets.

\begin{corollary}\label{th:ML-simple}
For any receiver $i$, if any packet is attempted $n$ times after time $t_i$ ($n\geq 1$), the sender can set this packet as not received if:
 \begin{equation}\label{eq:MLbound}
 \frac{\overline{p}_iq_i}{p_i} \leq 2^{1/n}-1 \;.
 \end{equation}
 Otherwise, it can be set as received. For $n=1$ (i.e. one attempt), the packet can directly be set as not received if \ignore{ $p_i > \frac{q_i}{1+q_i} > 0~\forall~i\in\mathcal{M}$} $p_i > q_i$.
\end{corollary}
\begin{proof}
From \eqref{eq:ML-Loss}, the ML state of a packet that is attempted $n$ times since $t_i$ can be computed as follows:
\begin{align*}
&\left(\frac{p_i}{p_i+\overline{p}_iq_i}\right)^n \geq 0.5 \\
\Rightarrow \quad & 2^{1/n}p_i \geq p_i+\overline{p}_iq_i  \\
\Rightarrow \quad & \frac{\overline{p}_iq_i}{p_i} \leq 2^{1/n}-1\;.
\end{align*}
In case of $n=1$, the decision rule becomes:
\begin{equation*}
\frac{\overline{p}_iq_i}{p_i} \leq 1 \quad \Rightarrow \quad q_i \leq \frac{p_i}{\overline{p}_i}\;.
\end{equation*}
If $p_i > q_i$, then $\frac{p_i}{\overline{p}_i} > q_i$ and the corollary follows.
\end{proof}

\begin{corollary}\label{th:ML-simple-reciprocal}
For reciprocal channels, the sender can set all attempted packets with $n=1$ for any receiver $i$ as not received, regardless of the value of $p_i$. For larger values of attempts (i.e. $n\geq 2$) for any packet, it can be set as not received if:
\begin{equation} \label{eq:MLbound-receiprocal}
 \overline{p}_i \leq 2^{1/n}-1\;,
\end{equation}
and set as received otherwise.
\end{corollary}
\begin{proof}
The first statement follows directly from the fact that $\frac{1}{2-p_i}$ is always greater than or equal to 0.5 for any value of $p_i$. The second statement follows directly by setting $q_i = p_i$ in \eqref{eq:MLbound}\ignore{ and re-arranging}.
\end{proof}

Clearly, the ML state estimation of the uncertain packets attempted only once requires no computations at all at the sender for both reciprocal channels and non-reciprocal channels with forward erasure probabilities larger than the reverse erasure probabilities (which is typically the case in wireless networks as explained earlier). For more than one attempt, we can see from both expressions in \eqref{eq:MLbound} and \eqref{eq:MLbound-receiprocal} that the left hand-side is a simple function of the erasure probabilities, which can be computed once per receiver. On the other hand, the right hand-side of both equations is a threshold function of $n$ only. Thus, the sender can construct an offline table for all the values of $2^{1/n}-1$\ignore{ $\left(\mbox{or}~2-2^{1/n}\right)$} for the general or reciprocal case, indexed by the variable $n\geq 2$, and store it in its memory. When an unheard feedback event occur from receiver $i$ for a given packet, the sender needs to compare the erasure probability function $\frac{\overline{p}_iq_i}{p_i}$ (or $\overline{p}_i$) in the general (or reciprocal) case to the entry in this table corresponding to the number of attempts of that packet, and determine the ML state estimate of this packet accordingly\ignore{to make the decision on updating the SFM and IDNC graph accordingly}.

Finally, it is important to note that the threshold function for both the non-reciprocal and reciprocal channel case is a monotonically decreasing function of $n$. Hence, once the threshold function drops below the erasure probability function for a given packet attempted $n$ times, which means that it is most likely to be received, the ML state for this uncertain packet will always remain the same for all future values of $n$ as long as $p_i$ and $q_i$ are not changed. Consequently, once the sender makes an estimate that a packet is most likely to be received, this estimate will never change and thus the sender need not compute it any further for this packet.

\subsection{Partially Blind Algorithm based on ML Graph Updates} \label{sec:blind-update}
Having derived the ML decision rules, we can design our proposed partially blind algorithm to solve the completion delay problem in lossy feedback environment. When the sender has uncertain packets in its SFM, due to unheard feedback events, it can employ the above ML state estimation rules to update the IDNC graph as follows. If the packet $j$ of vertex $v_{ij}$ is most likely to have been received at receiver $i$ (according to the ML rule derived in the previous section), vertex $v_{ij}$ is made hidden inside the IDNC graph, which means that it is temporarily not considered for transmission. Otherwise it is kept in the graph as an active vertex considered for any subsequent transmission. To minimize the completion delay, the sender selects the IDNC packet for each transmission by applying the maximum weight clique approach described in \sref{sec:perfect-formulation} on this estimated graph.

The resulting hidden vertices of any given receiver are treated according to what happens later:
\begin{itemize}
\item If the sender receives a feedback from this receiver, it will know its actual state and can update the status of these hidden vertices. Vertices representing received packets are eliminated and those representing lost packets are brought back as active (i.e. non-hidden) vertices in the graph.
\item If a receiver has no more active vertices in the graph but still has hidden ones, all these vertices are brought back as active vertices and are re-attempted within subsequent IDNC packets until a feedback is received from this receiver.\ignore{ If later a feedback comes from this receiver, before all these re-activated vertices are attempted, their actual status is updated similar to the above case (i.e. removed if received and kept as active if lost).}
\end{itemize}

As clarified above, the proposed solution for the completion delay problem is a maximum weight clique problem over the ML estimate of the IDNC graph. However, it is well-known that the maximum weight clique problem is NP-hard to solve and approximate \cite{Garey1979,Ausiello1999}. Consequently, we propose to use the same heuristic algorithm proposed in \cite{ICC10,TON10-CD} to solve the completion delay problem\ignore{ in perfect feedback scenarios}. This algorithm consists of a maximum weight vertex search algorithm, with the difference that the weight $\psi_i = \left(\frac{|\mathcal{W}_i|}{\overline{p}_i}\right)^m$ of each vertex $v_{ij}$ is modified to:
\begin{equation}
\widetilde{\psi}_{ij} = \psi_{i} \cdot \sum_{v_{kl}\in\mathcal{A}(v_{ij})} \psi_k\;,
\end{equation}
where $\mathcal{A}(v_{ij})$ is the set of vertices in the graph that are adjacent to vertex $v_{ij}$. Consequently, these new weights reflect, not only the weight of the vertex, but also the weights of the vertices adjacent to it.\ignore{ This is done using the following procedure:} Thus, when these modified weights are employed, each iteration of the maximum weight vertex search algorithm will select the vertex that has both high initial weight and strong adjacency to high initial weight vertices, which was shown to achieve near optimal performance in \cite{TON10-CD}. The complexity of this algorithm has been proved to be $O\left(M^2N\right)$ in \cite{TON10-CD}.

\section{Decoding Delay Problem} \label{sec:DD-extension}
In this section, we aim to study the effect of having feedback losses on the formulation of the decoding delay minimization problem in IDNC, and propose a new solution for this problem accordingly. As per its definition, the decoding delay increments occur after each transmission for the receivers that do not decode a new wanted packet from it. Consequently, we can derive an expression for the expected decoding delay increments in the presence of the uncertainties caused by unheard feedback occurrences, and then design a partially blind IDNC algorithm that minimizes this expression.\ignore{ Thus, we will start by deriving this expression in the next subsection.}

\subsection{Decoding Delay Expression in Lossy Feedback Environment}
Let $\kappa$ be the clique chosen for transmission at the current timeslot $t$, and let $d_i\left(\kappa\right)$ be the decoding delay increment experienced by receiver $i$ after this transmission. We also define the following sets:
\begin{itemize}
\item $\mathcal{O}$ is the set of outstanding receivers that are perceived by the sender to have non-empty Wants sets,
\item $\mathcal{F}\subseteq \mathcal{O}$ is the set of fully uncertain receivers, which includes any receiver $i$ having $\mathcal{W}_i\setminus\mathcal{U}_i = \emptyset$. In other words, this set includes any receiver $i$ that was targeted by the sender with all the packets remaining in its Wants set after time $t_i$ (i.e. the last time a feedback was heard from it).
\item $\nu(\kappa)\subseteq \mathcal{O}$ is the set of receivers that are not targeted with a primary packet in $\kappa$.
\item $\tau_n(\kappa)\subseteq \mathcal{O}$ is the set of receivers that are targeted in $\kappa$ with a new primary packet (i.e an un-attempted primary packet after their last heard feedback).
\item $\tau_u\left(\kappa\right)\subseteq \mathcal{O}$ is the set of receivers that are targeted in $\kappa$ with one of their uncertain primary packets (i.e previously attempted primary packets after the last heard feedback from these receivers).
\end{itemize}
Let $D\left(\kappa\right)$ be sum of all decoding delay increments of all receivers. Since only receivers in $\mathcal{O}$ may experience increments in decoding delay, we can write that $D\left(\kappa\right) = \sum_{i\in\mathcal{O}} d_i(\kappa)$. In the rest of this section, we only consider receivers in $\mathcal{O}$ in all derivations and formulae.

Defining $j_\kappa$ as the primary packet that receiver $i \in\left(\tau_n(\kappa)\cup\tau_u(\kappa)\right)$ can decode from the transmission of $\kappa$, the following theorem presents an expression for the expected sum decoding delay increments after this transmission.
\begin{theorem}\label{th:decoding-delay}
The expected sum decoding delay increment after the transmission of $\kappa$ at time $t>t_i~\forall~$ is:
\begin{equation}\label{eq:decoding-delay}
\mathds{E}[D(\kappa)] = \sum_{i\in \nu(\kappa)}\overline{p}_i
+ \sum_{i\in \tau_u(\kappa)}\left[\overline{p}_i\mathcal{P}_{ij_\kappa}^R\right] - \sum_{i \in \mathcal{F}} \left[\overline{p}_i\cdot \prod_{h\in\mathcal{W}_i}\mathcal{P}_{ih}^R\right]\;.
\ignore{%%%%%%%%%%%%%%%%%%
\mathbb{E}[\mathcal{D}(\kappa)] & =  \sum_{i \in \nu \left(\kappa\right)} \overline{p}_i - \sum_{i \in \mathcal{F}\cap \:\nu \left(\kappa\right)} \left[\overline{p}_i\cdot\prod_{j \in \mathcal{W}_i}\mathcal{P}^R_{ij}\right]  \nonumber\\
&+ \sum_{i \in \tau_u(\kappa)} \left[\overline{p}_i \mathcal{P}^R_{ij_\kappa}\right] - \sum_{i \in \mathcal{F}\cap\:\tau_u(\kappa)}\left[\overline{p}_i \mathcal{P}^R_{ij_\kappa}\cdot\prod_{j \in \mathcal{W}_i}\mathcal{P}^R_{ij}\right]\;.
}%%%%%%%%%%%%%%%%%%%%%%%%%%%%%%%%%%%%%%%%%%%%%%%%%%%%%%%%%%%%%%%%
\end{equation}
\end{theorem}
\begin{proof}
Since any $d_i(\kappa)$ in the definition of $D(\kappa)$ can be only either 0 or 1, we have:
\begin{singlecol}
\begin{equation}\label{eq:th2-proof1}
\mathds{E}[\mathcal{D}(\kappa)] = \sum_{i\in\mathcal{O}} \mathds{E}[d_i(\kappa)]
\ignore{& = \sum_{i\in\mathcal{O}}\left( 0 \times \mathds{P}\left\{d_i\left(\kappa\right) = 0\right\}  + 1 \times \mathds{P}\left\{d_i\left(\kappa\right) = 1\right\}\right)\nonumber\\}
  = \sum_{i\in\mathcal{O}} \mathds{P}\left\{d_i\left(\kappa\right) = 1\right\}\;.
\end{equation}
\end{singlecol}
\begin{doublecol}
\begin{align}\label{eq:th2-proof1}
\mathds{E}[\mathcal{D}(\kappa)] & = \sum_{i\in\mathcal{O}} \mathds{E}[d_i(\kappa)] \nonumber\\
\ignore{& = \sum_{i\in\mathcal{O}}\left( 0 \times \mathds{P}\left\{d_i\left(\kappa\right) = 0\right\}  + 1 \times \mathds{P}\left\{d_i\left(\kappa\right) = 1\right\}\right)\nonumber\\}
  & = \sum_{i\in\mathcal{O}} \mathds{P}\left\{d_i\left(\kappa\right) = 1\right\}\;.
\end{align}
\end{doublecol}
Consequently, to prove the theorem, we need to derive the probabilities of $d_i(\kappa) = 1$ for all the receivers in $\mathcal{O}$. According to the definition of decoding delay, and given the uncertainties arising from unheard feedback occurrences, any receiver $i$, perceived by the sender as having a non-empty Wants set, can be classified into one of two sets (again as perceived by the sender):\\%$\quad$\\
\textbf{Partially Uncertain Receivers}: $\overline{\mathcal{F}}=\left\{i\mid\mathcal{W}_i\setminus\mathcal{U}_i \neq \emptyset\right\}$\\
For any receiver $i$ in this set, there exist some packets in its Wants set that have never been attempted after $t_i$. For such receiver $i\in \overline{\mathcal{F}}$, we have the following possibilities:
\begin{itemize}
\item If this receiver $i \in \nu\left(\kappa\right)$, then it will have $d_i(\kappa) = 1$ if and only if it receives this transmission. Thus, we have:
\begin{equation}
\mathds{P}\left\{d_i\left(\kappa\right) = 1\right\}  = \overline{p}_i\;.
\end{equation}
\item If receiver $i \in \tau_n(\kappa)$ (i.e. $i$ is targeted by a primary packet that has not been attempted after $t_i$), then $i$ will never experience a decoding delay.
\item If receiver $i \in \tau_u\left(\kappa\right)$ (i.e. $i$ is targeted by a primary packet that has been attempted after $t_i$), then receiver $i$ will have $d_i(\kappa) = 1$ if and only if both following events are true:
\begin{itemize}
\item Receiver $i$ receives the transmission, which occurs with probability $\overline{p}_i$.
\item Receiver $i$ has received the attempted uncertain packet $j_\kappa$ in any one of its previous attempts after $t_i$, which occurs with probability $\mathcal{P}_{ij_\kappa}^R$.
\end{itemize}
Clearly, these two events are independent and we thus get:
\begin{equation}\label{eq:th2-proof2}
\mathds{P}\left\{d_i\left(\kappa\right) = 1\right\}  = \overline{p}_i \mathcal{P}_{ij_\kappa}^R\;. %\left(1-\mathcal{P}^{in}_{ij_\kappa}\right)\;.
\end{equation}
\ignore{where $\mathcal{P}^{in}_{ij_\kappa}$ is referred to as the innovation probability of packet $j$ at receiver $i$, defined as the probability that packet $j_\kappa$ was not previously received by this receiver in any previous attempts since $t_i$.}
\end{itemize}
%$\qquad$\\
\textbf{Fully Uncertain Receivers}: $\mathcal{F}=\left\{i\mid\mathcal{W}_i\setminus\mathcal{U}_i = \emptyset\right\}$\\
\ignore{As defined earlier,} This is the set of receivers for which the sender has previously attempted all their remaining wanted packets after their last heard feedback.
For any receiver $i \in \mathcal{F}$, we have the following possibilities:
\begin{singlecol}
\begin{itemize}
\item If this receiver $i \in \nu\left(\kappa\right)$, then it will have $d_i(\kappa) = 1$ if and only if both following events are true:
    \begin{itemize}
    \item Receiver $i$ receives this transmission.
     \item There exists at least one packet $h\in\mathcal{W}_i$, which was not received by receiver $i$ in all its attempts since $t_i$. This event occurs with probability $1-\prod_{h\in\mathcal{W}_i}\mathcal{P}_{ih}^R$.
     \end{itemize}
     Since these two events are independent, we thus get:
\begin{equation}
\mathds{P}\left\{d_i\left(\kappa\right) = 1\right\}  = \overline{p}_i\left(1-\prod_{h\in\mathcal{W}_i}\mathcal{P}_{ih}^R\right)\;.
\end{equation}
\ignore{where $\mathcal{P}^c_i$ is referred to as the completion probability of receiver $i$, defined as the probability that receiver $i$ has completed the reception of all its wanted packets in previous attempts from the sender.}
\item If receiver $i \in \tau_u\left(\kappa\right)$, then it will have $d_i(\kappa) = 1$  if and only if all the following three events are true:
   \begin{itemize}
   \item E1: Receiver $i$ receives the transmission.
   \item E2: There exists at least one packet $h\in\mathcal{W}_i$, which was not received by receiver $i$ in all its attempts after $t_i$.
   \item E3: Receiver $i$ has received the attempted uncertain packet $j_\kappa\in\mathcal{W}_i$ in any one of its previous attempts after $t_i$
   \end{itemize}
   Clearly, E1 is independent from both E2 and E3. However, E2 and E3 are dependent events because, for both E2 and E3 to be simultaneously true, $j_\kappa$ must not be among the packets that were not received by receiver $i$\ignore{ for E2 to be true}. Consequently, we need to find the joint probability of these two events. Both E2 and E3 will be simultaneously true if and only if receiver $i$ has received packet $j_\kappa$ in one of its previous attempts after $t_i$ and at least one other packet $h\in\mathcal{W}_i\setminus j_\kappa$ was not previously received in any previous attempts after $t_i$. Since $j_\kappa$ and $\mathcal{W}_i\setminus j_\kappa$ are mutually exclusive sets of packets, and since the reception/loss events of mutually exclusive sets of packets are independent, this joint probability can be mathematically written as:
   \begin{align}
 \mathds{P}\ignore{\{\mbox{Both E2,~E3 = True}\}}\{\mbox{E2}\wedge\mbox{E3}\} & = \mathds{P}\{j_\kappa~\mbox{received}\} \cdot \mathds{P}\{\exists\mbox{ at least one}~h\in\mathcal{W}_i\setminus j_\kappa~\ignore{\mbox{s.t}~h=}\mbox{missing}\} \nonumber \\
 & = \mathds{P}\{j_\kappa~\mbox{received}\}\cdot\left(1-\mathds{P}\{h~\mbox{received}~\forall~h\in\mathcal{W}_i\setminus j_\kappa\}\right)\nonumber\\
 & = \mathcal{P}_{ij_\kappa}^R\left(1-\prod_{h\in\mathcal{W}_i\setminus j_\kappa}\mathcal{P}_{ih}^R\right)\ignore{\nonumber\\
 &} = \mathcal{P}_{ij_\kappa}^R-\prod_{h\in\mathcal{W}_i}\mathcal{P}_{ih}^R\;.
 \end{align}
 Consequently, we get:
\begin{equation}
\mathds{P}\left\{d_i\left(\kappa\right) = 1\right\}  = \overline{p}_i\left(\mathcal{P}_{ij_\kappa}^R-\prod_{h\in\mathcal{W}_i}\mathcal{P}_{ih}^R\right)\;.
\end{equation}
\end{itemize}
\end{singlecol}
\begin{doublecol}
\begin{itemize}
\item If this receiver $i \in \nu\left(\kappa\right)$, then it will have $d_i(\kappa) = 1$ if and only if both following events are true:
    \begin{itemize}
    \item Receiver $i$ receives this transmission.
     \item There exists at least one packet $h\in\mathcal{W}_i$, which was not received by receiver $i$ in all its attempts since $t_i$. This event occurs with probability $1-\prod_{h\in\mathcal{W}_i}\mathcal{P}_{ih}^R$.
     \end{itemize}
     Since these two events are independent, we thus get:
\begin{equation}
\mathds{P}\left\{d_i\left(\kappa\right) = 1\right\}  = \overline{p}_i\left(1-\prod_{h\in\mathcal{W}_i}\mathcal{P}_{ih}^R\right)\;.
\end{equation}
\ignore{where $\mathcal{P}^c_i$ is referred to as the completion probability of receiver $i$, defined as the probability that receiver $i$ has completed the reception of all its wanted packets in previous attempts from the sender.}
\item If receiver $i \in \tau_u\left(\kappa\right)$, then it will have $d_i(\kappa) = 1$  if and only if all the following three events are true:
   \begin{itemize}
   \item E1: Receiver $i$ receives the transmission.
   \item E2: There exists at least one packet $h\in\mathcal{W}_i$, which was not received by receiver $i$ in all its attempts after $t_i$.
   \item E3: Receiver $i$ has received the attempted uncertain packet $j_\kappa\in\mathcal{W}_i$ in any one of its previous attempts after $t_i$
   \end{itemize}
   Clearly, E1 is independent from both E2 and E3. However, E2 and E3 are dependent events because, for both E2 and E3 to be simultaneously true, $j_\kappa$ must not be among the packets that were not received by receiver $i$\ignore{ for E2 to be true}. Consequently, we need to find the joint probability of these two events. Both E2 and E3 will be simultaneously true if and only if receiver $i$ has received packet $j_\kappa$ in one of its previous attempts after $t_i$ and at least one other packet $h\in\mathcal{W}_i\setminus j_\kappa$ was not previously received in any previous attempts after $t_i$. Since $j_\kappa$ and $\mathcal{W}_i\setminus j_\kappa$ are mutually exclusive sets of packets, and since the reception/loss events of mutually exclusive sets of packets are independent, this joint probability can be mathematically written as:
   \begin{align}
 \mathds{P}\{\mbox{E2}\wedge\mbox{E3}\} & = \mathds{P}\{j_\kappa~\mbox{received}\} \nonumber\\
  &\quad\times \mathds{P}\{\exists\mbox{ at least one}~h\in\mathcal{W}_i\setminus j_\kappa~\ignore{\mbox{s.t}~h=}\mbox{missing}\} \nonumber \\
 & = \mathds{P}\{j_\kappa~\mbox{received}\}\nonumber\\
 & \quad \times \left(1-\mathds{P}\{h~\mbox{received}~\forall~h\in\mathcal{W}_i\setminus j_\kappa\}\right)\nonumber\\
 & = \mathcal{P}_{ij_\kappa}^R\left(1-\prod_{h\in\mathcal{W}_i\setminus j_\kappa}\mathcal{P}_{ih}^R\right)\nonumber\\
 & = \mathcal{P}_{ij_\kappa}^R-\prod_{h\in\mathcal{W}_i}\mathcal{P}_{ih}^R\;.
 \end{align}

 Consequently, we get:
\begin{equation}
\mathds{P}\left\{d_i\left(\kappa\right) = 1\right\}  = \overline{p}_i\left(\mathcal{P}_{ij_\kappa}^R-\prod_{h\in\mathcal{W}_i}\mathcal{P}_{ih}^R\right)\;.
\end{equation}
\end{itemize}
\end{doublecol}
Given all the previous cases, we can\ignore{ summarize the probabilities of decoding delay increments equal to 1 as follows} express $\mathds{P}\left\{d_i\left(\kappa\right) = 1 \right\}$ as:
\begin{equation}\label{eq:th2-proof3}
\mathds{P}\left\{d_i\left(\kappa\right) = 1 \right\} =
\begin{cases}
%0    \quad & i \notin \mathcal{O} \\
\overline{p}_i     & i \in \overline{\mathcal{F}}\cap\nu(\kappa) \\
\overline{p}_i\left(1-\prod_{h\in\mathcal{W}_i}\mathcal{P}_{ih}^R\right)     & i \in \mathcal{F}\cap\nu(\kappa) \\
0    & i \in \overline{\mathcal{F}}\cap\tau_n(\kappa) \\
\overline{p}_i \mathcal{P}_{ij_\kappa}^R    & i\in \overline{\mathcal{F}}\cap\tau_u(\kappa) \\
\overline{p}_i\left(\mathcal{P}_{ij_\kappa}^R-\prod_{h\in\mathcal{W}_i}\mathcal{P}_{ih}^R\right)   & i \in \mathcal{F}\cap\tau_u(\kappa)\:.
\end{cases}
\end{equation}
Substituting \eqref{eq:th2-proof3} in \eqref{eq:th2-proof1}, we get:
\begin{singlecol}
\begin{align}
\mathds{E}[D(\kappa)] &= \sum_{i\in \overline{\mathcal{F}}\cap\:\nu(\kappa)}\overline{p}_i +  \sum_{i\in \mathcal{F}\cap\:\nu(\kappa)} \left[\overline{p}_i\left(1-\prod_{h\in\mathcal{W}_i}\mathcal{P}_{ih}^R\right)\right] \nonumber\\
& + \sum_{i\in \overline{\mathcal{F}}\cap\:\tau_u(\kappa)}\left[\overline{p}_i \mathcal{P}_{ij_\kappa}^R\right] + \sum_{i \in \mathcal{F}\cap\:\tau_u(\kappa)} \left[\overline{p}_i\left(\mathcal{P}_{ij_\kappa}^R-\prod_{h\in\mathcal{W}_i}\mathcal{P}_{ih}^R\right)\right]\;.
\end{align}
\end{singlecol}
\begin{doublecol}
\begin{align}
\mathds{E}[D(\kappa)] &= \sum_{i\in \overline{\mathcal{F}}\cap\:\nu(\kappa)}\overline{p}_i +  \sum_{i\in \mathcal{F}\cap\:\nu(\kappa)} \left[\overline{p}_i\left(1-\prod_{h\in\mathcal{W}_i}\mathcal{P}_{ih}^R\right)\right] \nonumber\\
& + \sum_{i\in \overline{\mathcal{F}}\cap\:\tau_u(\kappa)}\left[\overline{p}_i \mathcal{P}_{ij_\kappa}^R\right] \nonumber \\
&+ \sum_{i \in \mathcal{F}\cap\:\tau_u(\kappa)} \left[\overline{p}_i\left(\mathcal{P}_{ij_\kappa}^R-\prod_{h\in\mathcal{W}_i}\mathcal{P}_{ih}^R\right)\right]\;.
\end{align}
\end{doublecol}
Now, since $\left(\overline{\mathcal{F}}\cap~\nu(\kappa)\right) \cup \left(\mathcal{F}\cap~\nu(\kappa)\right) = \nu(\kappa)$ and since $\left(\overline{\mathcal{F}}\cap~\tau_u(\kappa)\right) \cup \left(\mathcal{F}\cap~\tau_u(\kappa)\right) = \tau_u(\kappa)$, we can use these facts to group similar terms, which reduces the expression to:
\begin{align}\label{eq:th2-proof4}
\mathds{E}[D(\kappa)] &= \sum_{i\in \nu(\kappa)}\overline{p}_i -  \sum_{i\in \mathcal{F}\cap\:\nu(\kappa)} \left[\overline{p}_i\cdot\prod_{h\in\mathcal{W}_i}\mathcal{P}_{ih}^R\right] \nonumber\\
& + \sum_{i\in \tau_u(\kappa)}\left[\overline{p}_i\mathcal{P}_{ij_\kappa}^R\right] - \sum_{i \in \mathcal{F}\cap\:\tau_u(\kappa)} \left[\overline{p}_i\cdot \prod_{h\in\mathcal{W}_i}\mathcal{P}_{ih}^R\right]\;.
\end{align}
Finally, since $\left(\mathcal{F}\cap\:\nu(\kappa)\right) \cup \left(\mathcal{F}\cap\:\tau_u(\kappa)\right) = \mathcal{F}$, we get:

\begin{equation}\label{eq:th2-proof5}
\mathds{E}[D(\kappa)] = \sum_{i\in \nu(\kappa)}\overline{p}_i
+ \sum_{i\in \tau_u(\kappa)}\left[\overline{p}_i\mathcal{P}_{ij_\kappa}^R\right] - \sum_{i \in \mathcal{F}} \left[\overline{p}_i\cdot \prod_{h\in\mathcal{W}_i}\mathcal{P}_{ih}^R\right]\;.
\end{equation}
\end{proof}

The right-hand side of the sum decoding delay expression in \eqref{eq:decoding-delay} can be intuitively interpreted as follows:
\begin{itemize}
\item The first summation represents the expected decoding delay increments of non-targeted receivers, when each receiver $i\in\nu(\kappa)$ receives this transmission $\left(\mbox{with probability}~\overline{p}_i\right)$.
\item The second summation represents the expected decoding delay increments for targeted receivers with uncertain packets, when each receiver $i \in \tau_u(\kappa)$ receives this transmission $\left(\mbox{with probability}~\overline{p}_i\right)$ and it turns out that it has previously received the packet $j_\kappa$ in one or more of its previous attempts since time $t_i$ $\left(\mbox{with probability}~\mathcal{P}^R_{ij_\kappa}\right)$.
\item The last negative summation represents the fact that the decoding delays expressed in the previous two summations will not be experienced by both targeted and non-targeted receivers in $\mathcal{F}$ (i.e. receivers for which the sender has previously attempted all their remaining wanted packets after their last heard feedback), even if they receive this transmission $\left(\mbox{with probability}~\overline{p}_i~\mbox{for receiver}~i\right)$, in case they have received all of these uncertain packets in one or more of their previous attempts $\left(\mbox{with probability}~\prod_{h\in\mathcal{W}_i}\mathcal{P}_{ih}^R~\mbox{for receiver}~i\right)$.
\end{itemize}

\subsection{Decoding Delay Problem Formulation}
Given the previous result, we can now formulate the minimum decoding delay problem in lossy feedback environment in the following corollary.
\begin{corollary} \label{th:DD-formulation}
The problem of selecting an IDNC coded packet for a given transmission, which would result in the minimum expected sum decoding delay increments\ignore{ after this transmission}, is equivalent to a maximum weight clique problem over the IDNC graph, in which the weight $\omega_{ij}$ of every vertex $v_{ij}$ is set as follows:
\begin{equation}\label{eq:weights}
\omega_{ij} =
\begin{cases}
\overline{p}_i \quad & j\in\mathcal{W}_i\setminus\mathcal{U}_i\\
\overline{p}_i\left(1-\mathcal{P}^R_{ij}\right) \quad & j \in \mathcal{W}_i\cap\mathcal{U}_i \;.
\end{cases}
\end{equation}
\end{corollary}
\begin{proof}
From the derived expression \eqref{eq:decoding-delay} in \thref{th:decoding-delay}, the problem of selecting a clique $\kappa^*$ for a given transmission, so as to minimize the expected sum decoding delay increments after this transmission, can be expressed as:
\begin{singlecol}
\begin{align}
\kappa^* & = \arg\min_{\kappa\in\mathcal{G}} \left\{\mathds{E}\left[D(\kappa)\right]\right\} \nonumber\\
& = \arg\min_{\kappa\in\mathcal{G}} \left\{\sum_{i\in \nu(\kappa)}\overline{p}_i
+ \sum_{i\in \tau_u(\kappa)}\left[\overline{p}_i\mathcal{P}_{ij_\kappa}^R\right] - \sum_{i \in \mathcal{F}} \left[\overline{p}_i\cdot \prod_{h\in\mathcal{W}_i}\mathcal{P}_{ih}^R\right]\right\} \nonumber \\
& = \arg\max_{\kappa\in\mathcal{G}} \left\{ - \sum_{i\in \nu(\kappa)}\overline{p}_i
- \sum_{i\in \tau_u(\kappa)}\left[\overline{p}_i\mathcal{P}_{ij_\kappa}^R\right] + \sum_{i \in \mathcal{F}} \left[\overline{p}_i\cdot \prod_{h\in\mathcal{W}_i}\mathcal{P}_{ih}^R\right]\right\}\;.
\end{align}
\end{singlecol}
\begin{doublecol}
\begin{align}
\kappa^* & = \arg\min_{\kappa\in\mathcal{G}} \left\{\mathds{E}\left[D(\kappa)\right]\right\} \nonumber\\
& = \arg\min_{\kappa\in\mathcal{G}} \Bigg\{\sum_{i\in \nu(\kappa)}\overline{p}_i
+ \sum_{i\in \tau_u(\kappa)}\left[\overline{p}_i\mathcal{P}_{ij_\kappa}^R\right] \nonumber\\
& \qquad \qquad \qquad \qquad \quad - \sum_{i \in \mathcal{F}} \left[\overline{p}_i\cdot \prod_{h\in\mathcal{W}_i}\mathcal{P}_{ih}^R\right]\Bigg\} \nonumber \\
& = \arg\max_{\kappa\in\mathcal{G}} \Bigg\{ - \sum_{i\in \nu(\kappa)}\overline{p}_i
- \sum_{i\in \tau_u(\kappa)}\left[\overline{p}_i\mathcal{P}_{ij_\kappa}^R\right] \nonumber\\
& \qquad \qquad \qquad \qquad \qquad ~ + \sum_{i \in \mathcal{F}} \left[\overline{p}_i\cdot \prod_{h\in\mathcal{W}_i}\mathcal{P}_{ih}^R\right]\Bigg\}\;.
\end{align}
\end{doublecol}
The last summation term inside the braces is not dependent on $\kappa$, as it affects all receivers in $\mathcal{F}$ whether selected in $\kappa$ or not. Consequently, we can remove it from the braces without affecting the maximization problem. Similarly, adding inside the braces another term that is independent on $\kappa$, such as $\sum_{i\in\mathcal{M}} \overline{p}_i$, will not affect the result of the maximization problem. But since $\mathcal{M}\setminus \nu(\kappa) = \tau_n(\kappa)\cup \tau_u(\kappa)$, we get:
\begin{align}\label{eq:cl-proof1}
\kappa^* & = \arg\max_{\kappa\in\mathcal{G}} \left\{\sum_{i\in\mathcal{M}} \overline{p}_i - \sum_{i\in \nu(\kappa)}\overline{p}_i
- \sum_{i\in \tau_u(\kappa)}\left[\overline{p}_i\mathcal{P}_{ij_\kappa}^R\right] \right\} \nonumber \\
& = \arg\max_{\kappa\in\mathcal{G}} \left\{\sum_{i\in \left(\tau_n(\kappa) \cup\: \tau_u(\kappa)\right)} \overline{p}_i
- \sum_{i\in \tau_u(\kappa)}\left[\overline{p}_i\mathcal{P}_{ij_\kappa}^R\right] \right\} \nonumber\\
& = \arg\max_{\kappa\in\mathcal{G}} \left\{\sum_{i\in \tau_n(\kappa)} \overline{p}_i
+ \sum_{i\in \tau_u(\kappa)}\left[\overline{p}_i\left(1-\mathcal{P}_{ij_\kappa}^R\right)\right] \right\} \;.
\end{align}
Clearly, the above expression in \eqref{eq:cl-proof1} is equivalent to the maximum weight clique problem over the IDNC graph given the weights $\omega_{ij}$ defined in \eqref{eq:weights}.
\end{proof}

\subsection{Low Complexity Algorithm}
Similar to the completion delay problem, the proposed solution for the decoding delay problem is a maximum weight clique problem over the IDNC graph, where all uncertain vertices are kept but weighted differently from the un-attempted vertices as shown in \eqref{eq:weights}. But again, this problem is NP-hard to solve \cite{Garey1979} and NP-hard to approximate \cite{Ausiello1999}. Consequently, we propose to use the same heuristic algorithm proposed in \cite{GC10} to solve the decoding delay problem\ignore{ in perfect feedback scenarios}. This algorithm is similar to the one proposed for the completion delay problem above, but uses the new weights defined in Corollary \ref{th:DD-formulation}. It consists of a weighted vertex search (WVS) algorithm, with modified weights defined as:
\begin{equation}
\widetilde{\omega}_{ij} = \omega_{ij} \cdot \sum_{v_{kl}\in\mathcal{A}(v_{ij})} \omega_{kl}\;,
\end{equation}
where $\mathcal{A}(v_{ij})$ is the set of vertices adjacent to vertex $v_{ij}$. Consequently, the algorithm selects in each iteration the vertex that not only has the highest weight $\omega_{ij}$ but also has strong adjacency to other high weight vertices in the graph. Similar to the completion delay heuristic, the complexity of this algorithm is $O\left(M^2N\right)$, and it was shown to achieve a small degradation compared to the complex optimal maximum weight clique solution \cite{GC10}.

\section{Simulation Results} \label{sec:simulations}
In this section, we compare, through extensive simulations, the performances of our proposed solutions for the completion and decoding delay problems, in lossy feedback environments, to both the perfect feedback (PF) scenario and two other blind graph update approaches described in \cite{ICC11}:
\begin{itemize}
\item Full Vertex Elimination (FVE): Each attempted vertex with unheard feedback is assumed to be hidden in the graph and is temporarily not considered for transmission. These hidden vertices are treated later similar to our proposed algorithm in \sref{sec:blind-update}.
\item No Vertex Elimination (NVE): Each attempted vertex with unheard feedback is kept in the graph.\ignore{ NVE rapidly re-attempts these uncertain vertices, thus giving the chance to the sender to receive feedback from their receivers and to determine their accurate reception status. This fast re-attempts of vertices and SFM update may be of greater importance in this lossy feedback context, especially at the end of the recovery phase, as it can help the sender make better coding decisions towards completion.}
\end{itemize}
We assume channel reciprocity and set $n=3$ in the vertex weight. We also assume that the different receivers experience different erasure probabilities and different demand ratios, while maintaining the average erasure probability $p$ and the average demand ratio $\mu$ constant for each simulation point.

\subsection{Completion Delay Results}
Figures \ref{fig:P}, \ref{fig:D}, \ref{fig:M} and \ref{fig:N} depict the comparison of the average completion delay, achieved by the different algorithms, against $p$ (for $\mu=0.5$, $M=60$, $N=30$), $\mu$ (for $M=60$, $N=30$, $p=0.15$), $M$ (for $\mu=0.5$, $N=30$ and $p=\{0.15,0.5\}$), and $N$ (for $\mu=0.5$, $M=60$, and $p=\{0.15,0.5\}$), respectively.
\begin{figure*}[t]
\centering
    \begin{subfigure}[b]{0.5\textwidth}
    \centering
    \includegraphics[width=1\linewidth]{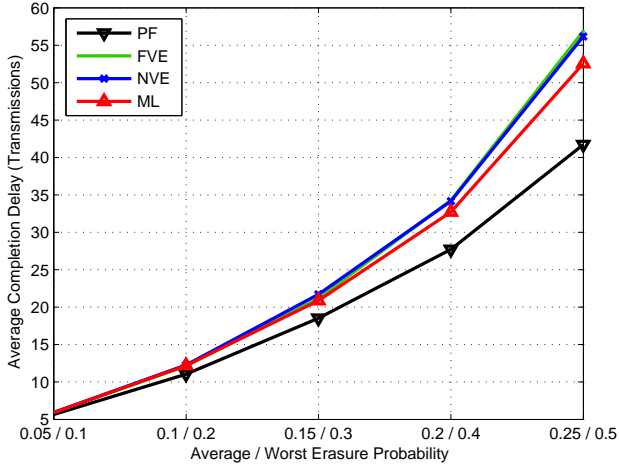}
    \caption{Average completion delay vs $p$.} \label{fig:P}
    \end{subfigure}~
    \begin{subfigure}[b]{0.5\textwidth}
    \centering
    \includegraphics[width=1\linewidth]{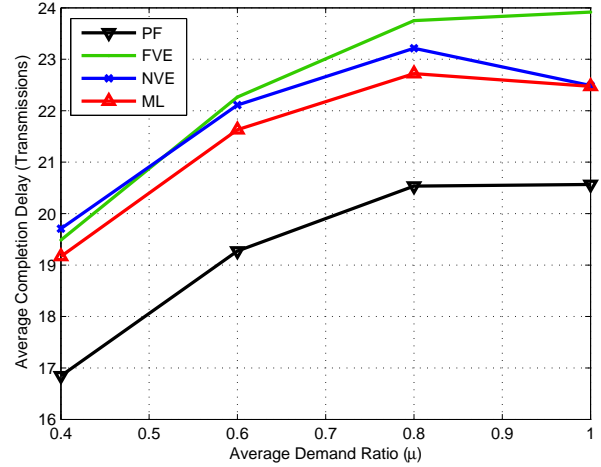}
    \caption{Average completion delay vs $\mu$.} \label{fig:D}
    \end{subfigure}

    \begin{subfigure}[b]{0.5\textwidth}
    \centering
    \includegraphics[width=1\linewidth]{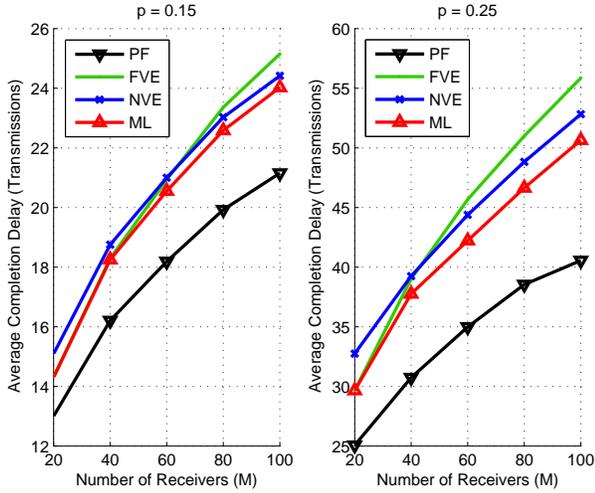}
    \caption{Average completion delay vs $M$.} \label{fig:M}
    \end{subfigure}~
    \begin{subfigure}[b]{0.5\textwidth}
    \centering
    \includegraphics[width=1\linewidth]{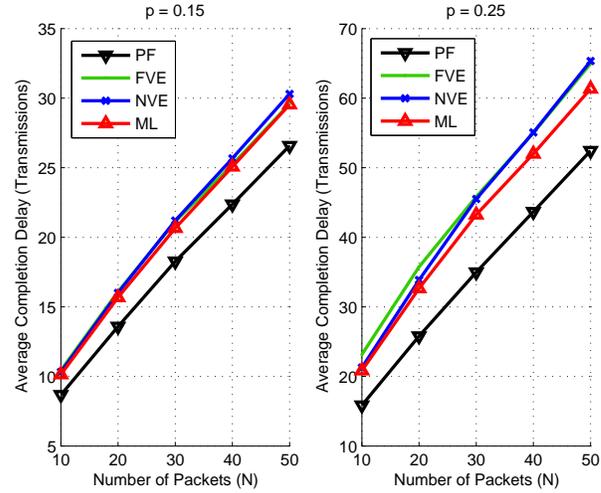}
    \caption{Average completion delay vs $N$.} \label{fig:N}
    \end{subfigure}
     \caption{Completion delay results}
\end{figure*}

From all figures, we can observe the better performance of the ML approach over the FVE and NVE approaches in reducing the average completion delay, especially for receivers with high erasure probabilities and also when $\mu<1$. The result conforms with the analysis in \sref{sec:belief}, showing that the ML update approach of the graph represents the ML estimation of the actual graph, and thus the coding decisions deduced from this graph would on average achieve the best results compared to other graph update approaches.\ignore{ Moreover, we see that the average completion delay of (ML) is not greatly affected by both the non-innovative transmissions of received packets and the potential estimation errors, as long as it always follows the ML state of the network.}

We can also observe that FVE outperforms NVE and approaches ML performance at low demand ratios and small number of receivers, whereas NVE outperforms FVE and approaches ML performance at high demand ratios and large number of receivers. This can be explained from the characteristics of the two approaches. FVE re-attempts the remaining unacknowledged vertices only when all the graph vertices are blindly depleted. Since the primary graph sizes are relatively small for small numbers of receivers and low demand ratios, FVE tends to blindly deplete its graphs very fast and start the re-attempting of the uncertain vertices early. This allows it to finish faster than NVE, which re-attempts vertices a lot in low demand environment.

At large numbers of receivers and high demand ratios, the larger size of the IDNC primary graph increases the time for the FVE approach to blindly deplete the graph. Consequently, all receivers with remaining uncertain vertices will have to wait longer for FVE to re-attempt them, which results in a larger completion delay. This effect increases for the receivers with smaller Wants sets. On the other hand, FVE reduces this effect since it leaves these unacknowledged vertices in the graph, which gives them a chance to speed up their transmission re-attempt, their recovery and/or their feedback. Most importantly, we can see that the ML algorithm achieves a better or similar performance than both of them in all cases, without much addition in complexity.

Finally, we can observe a degradation in the average completion delay obtained in the LF environment, compared to the PF environment. However, for a relatively large network setting ($M=100$, $N=30$), an average erasure probability of 0.25, and worst erasure probability of $0.5$ (\fref{fig:M}), this degradation in the frame delivery duration (from the start of the frame transmission until its reception at all receivers) reaches $14\%$, compared to the perfect feedback algorithm performance. For as high as 0.5 worst-case packet loss probability (\fref{fig:P}), we obtain a degradation of $14\%$. These values are clearly tolerable in such very large networks and up to 50$\%$ loss rate of feedback, which is typically very high for signalling information.

\subsection{Decoding Delay Results}
Figures \ref{fig:PP}, \ref{fig:DD}, \ref{fig:MM} and \ref{fig:NN} depict the comparison of the average decoding delay, achieved by the different algorithms, against $p$ (for $\mu=0.8$, $M=60$, $N=30$), $\mu$ (for $M=60$, $N=30$, $p=0.25$), $M$ (for $\mu=0.8$, $N=30$ and $p=\{0.25,0.5\}$), and $N$ (for $\mu=0.8$, $M=60$, and $p=\{0.25,0.5\}$), respectively.

\begin{figure*}[t]
\centering
    \begin{subfigure}[b]{0.5\textwidth}
    \centering
    \includegraphics[width=1\linewidth]{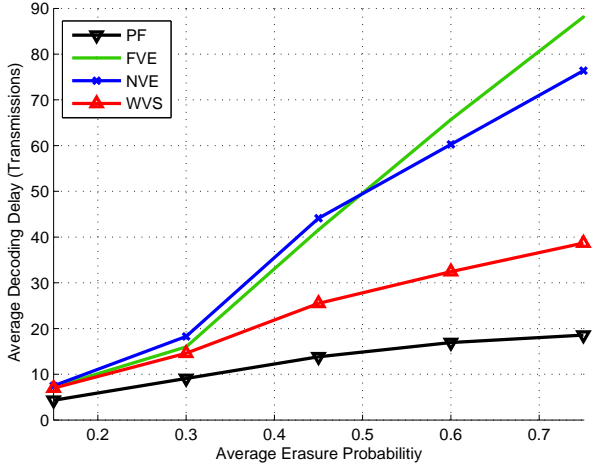}
    \caption{Average decoding delay vs $p$.} \label{fig:PP}
    \end{subfigure}~
    \begin{subfigure}[b]{0.5\textwidth}
    \centering
    \includegraphics[width=1\linewidth]{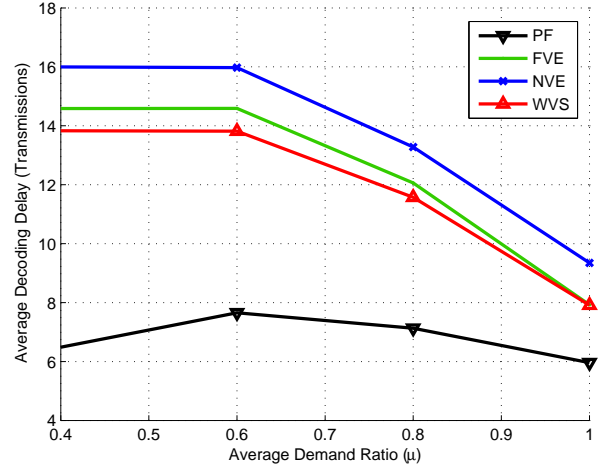}
    \caption{Average decoding delay vs $\mu$.} \label{fig:DD}
    \end{subfigure}

    \begin{subfigure}[b]{0.5\textwidth}
    \centering
    \includegraphics[width=1\linewidth]{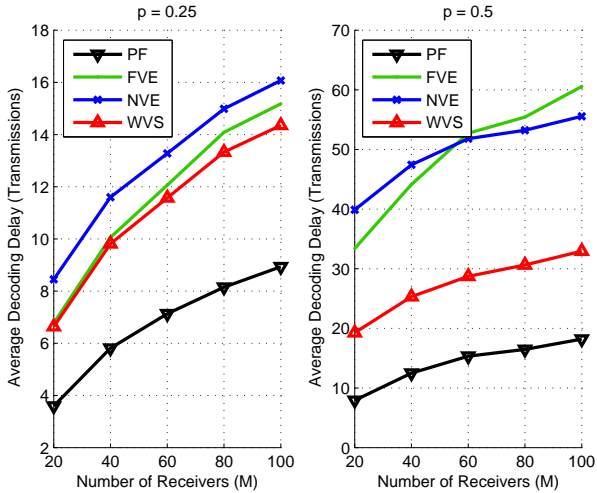}
    \caption{Average decoding delay vs $M$.} \label{fig:MM}
    \end{subfigure}~
    \begin{subfigure}[b]{0.5\textwidth}
    \centering
    \includegraphics[width=1\linewidth]{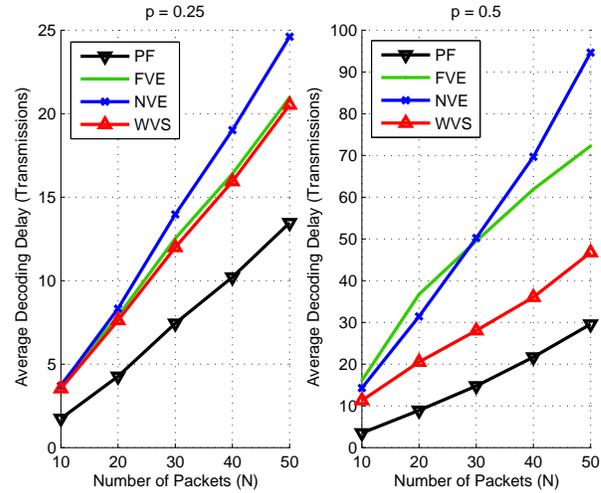}
    \caption{Average decoding delay vs $N$.} \label{fig:NN}
    \end{subfigure}
    \caption{Decoding delay results}
\end{figure*}

From all figures, we can clearly observe that, for all ranges of the different metrics, our proposed WVS approach, which gives different weights to unattempted and uncertain vertices, outperforms both NVE and FVE, which strictly keeps and eliminates the uncertain vertices, respectively.

We can also observe that FVE performs better than NVE in most scenarios, except for large erasure probabilities. This can be interpreted as follows. In most cases, it is usually better for reducing the decoding delay to attempt new vertices rather than attempted uncertain ones, and keep these uncertain ones either for later or for feedback from the same receivers on new attempted packets. Since FVE eliminates such uncertain vertices for any given receiver as long as there are new ones to attempt for this receiver, it thus outperforms NVE, which keeps the uncertain vertices in the process.

However, when the erasure probability becomes larger, it becomes very likely that the uncertain vertices are actually lost, which makes NVE a much more likely candidate of reflecting the actual state of the network compared to FVE.\ignore{ follows a similar ML trend as the one proposed for the completion delay problem\footnote{The ML here is in the forward link sense only (without considering what happened to the feedback), because the decoding delay is affected by individual packet losses from each individual receiver and not the global state of the network. Consequently, ML do not follow the bounds derived in Corollaries \ref{th:ML-simple} and \ref{th:ML-simple-reciprocal} but rather the simple ML rule that the packet is received (lost) on the forward link if the erasure probability is smaller (larger) than 0.5}, whereas FVE falls away from this ML trend.} That is why we can observe in \fref{fig:PP} that the larger the erasure probability, the better the performance of NVE compared to FVE. This also appears in some scenarios in Figures \ref{fig:MM} and \ref{fig:NN}, when $p = 0.5$.

\section{Conclusion} \label{sec:conclusion}
In this paper, we extended the study of the multicast completion and decoding delay minimization problems of IDNC to the lossy feedback environment. For the completion delay problem, we first extended the SSP formulation in perfect feedback to a POSSP formulation, reflecting the uncertainties resulting from unheard feedback events in the lossy feedback environment. We then used this formulation to identify the ML state of the network in events of unheard feedback, and employed it to design a partially blind graph update extension to the multicast IDNC algorithm in \cite{TON10-CD}. For the decoding delay problem, we derived an expression for the expected sum decoding delay increments after any arbitrary transmission and used it to find the optimal policy to reduce the decoding delay in such lossy feedback environment. Simulation results showed that our proposed partially blind solutions both outperformed other approaches and achieved a tolerable degradation in performance (compared to the perfect feedback environment) even at relatively high feedback loss rates.

\ignore{In this letter, we extended the problem of minimizing the completion delay of IDNC to LF environment. We first extended our SSP formulation for the PF environment in \cite{TON10-CD} to a POSSP formulation in the LF environment, by defining the belief state that reflects the uncertainty effects resulting from unheard feedback events. From this extended formulation, we derived the ML estimation of the network state in case of unheard feedback events from any number of receivers. We then proposed a partially blind IDNC graph extension to the algorithms in \cite{TON10-CD}, which follows the ML state of the network. Simulation results showed that this ML graph update approach outperforms two other update approaches for all parameter ranges, especially for large number of receivers. Moreover, the proposed algorithm can achieve a tolerable degradation, compared to the PF environment performance, for considerably high feedback loss rates.}

\bibliographystyle{IEEEtran}
\bibliography{IEEEabrv,bibfile}

\begin{thebibliography}{10}
\providecommand{\url}[1]{#1}
\csname url@samestyle\endcsname
\providecommand{\newblock}{\relax}
\providecommand{\bibinfo}[2]{#2}
\providecommand{\BIBentrySTDinterwordspacing}{\spaceskip=0pt\relax}
\providecommand{\BIBentryALTinterwordstretchfactor}{4}
\providecommand{\BIBentryALTinterwordspacing}{\spaceskip=\fontdimen2\font plus
\BIBentryALTinterwordstretchfactor\fontdimen3\font minus
  \fontdimen4\font\relax}
\providecommand{\BIBforeignlanguage}[2]{{%
\expandafter\ifx\csname l@#1\endcsname\relax
\typeout{** WARNING: IEEEtran.bst: No hyphenation pattern has been}%
\typeout{** loaded for the language `#1'. Using the pattern for}%
\typeout{** the default language instead.}%
\else
\language=\csname l@#1\endcsname
\fi
#2}}
\providecommand{\BIBdecl}{\relax}
\BIBdecl

\bibitem{PIMRC11}
S.~Sorour and S.~Valaee, ``Completion delay reduction in lossy feedback
  scenarios for instantly decodable network coding,'' \emph{IEEE International
  Symposium on Personal, Indoor and Mobile Radio Communications (PIMRC'11)},
  2011.

\bibitem{GC10}
------, ``Minimum broadcast decoding delay for generalized instantly decodable
  network coding,'' \emph{IEEE Global Telecommunications Conference
  ({GLOBECOM'10})}, Dec. 2010.

\bibitem{TON10-CD}
\BIBentryALTinterwordspacing
------, ``Completion delay minimization for instantly decodable network
  codes,'' \emph{submitted to IEEE/ACM Transactions on Networking}. [Online].
  Available: \url{http://arxiv.org/submit/402828/view}
\BIBentrySTDinterwordspacing

\bibitem{4895447}
J.~Sundararajan, D.~Shah, and M.~Medard, ``Online network coding for optimal
  throughput and delay - {T}he three-receiver case,'' \emph{International
  Symposium on Information Theory and Its Applications (ISITA'08)}, Dec. 2008.

\bibitem{4476183}
L.~Keller, E.~Drinea, and C.~Fragouli, ``Online broadcasting with network
  coding,'' \emph{Fourth Workshop on Network Coding, Theory and Applications
  (NetCod'08)}, Jan. 2008.

\bibitem{Drinea2009}
E.~Drinea, C.~Fragouli, and L.~Keller, ``Delay with network coding and
  feedback,'' \emph{IEEE International Symposium on Information Theory
  (ISIT'09)}, pp. 844--848, Jun. 2009.

\bibitem{4397057}
D.~Nguyen, T.~Nguyen, and X.~Yang, ``Multimedia wireless transmission with
  network coding,'' \emph{Packet Video Workshop (PV'07)}, pp. 326--335, Nov.
  2007.

\bibitem{5152148}
D.~Nguyen and T.~Nguyen, ``Network coding-based wireless media transmission
  using {POMDP},'' \emph{Packet Video Workshop (PV'09)}, May 2009.

\bibitem{ICC10}
S.~Sorour and S.~Valaee, ``On minimizing broadcast completion delay for
  instantly decodable network coding,'' \emph{IEEE International Conference on
  Communications (ICC'10)}, May 2010.

\bibitem{Sadeghi2010}
P.~Sadeghi, R.~Shams, and D.~Traskov, ``An optimal adaptive network coding
  scheme for minimizing decoding delay in broadcast erasure channels,''
  \emph{EURASIP Journal of Wireless Communications and Networking}, vol. 2010,
  pp. 1--14, Apr. 2010.

\bibitem{5425315}
D.~Traskov, M.~Medard, P.~Sadeghi, and R.~Koetter, ``Joint scheduling and
  instantaneously decodable network coding,'' \emph{Global Telecommunications
  Conference (GLOBECOM'09)}, Dec. 2009.

\bibitem{Li2011}
X.~Li, C.-C. Wang, and X.~Lin;, ``On the capacity of immediately-decodable
  coding schemes for wireless stored-video broadcast with hard deadline
  constraints,'' \emph{IEEE Journal on Selected Areas in Communications},
  vol.~29, no.~5, pp. 1094--1105, May 2011.

\bibitem{6030131}
C.~Zhan, V.~Lee, J.~Wang, and Y.~Xu, ``Coding-based data broadcast scheduling
  in on-demand broadcast,'' \emph{IEEE Transactions on Wireless
  Communications}, vol.~10, no.~11, pp. 3774 --3783, Nov. 2011.

\bibitem{ISIT09}
S.~Sorour and S.~Valaee, ``Adaptive network coded retransmission scheme for
  wireless multicast,'' \emph{IEEE International Symposium on Information
  Theory (ISIT'09)}, pp. 2577 -- 2581, Jun. 2009.

\bibitem{Patek99onpartially}
S.~Patek, ``On partially observed stochastic shortest path problems,''
  \emph{40th IEEE Conference on Decision and Control (CDC'01)}, pp. 5050--5055,
  1999.

\bibitem{Garey1979}
M.~Garey and D.~Johonson, \emph{{Computers and Intractability - A Guide to the
  Theory of NP-Completeness}}.\hskip 1em plus 0.5em minus 0.4em\relax Freeman,
  New York, 1979.

\bibitem{Ausiello1999}
G.~Ausiello, P.~Crescenzi, G.~Gambosi, V.~Kann, A.~Marchetti-Spaccamela, and
  M.~Protasi, \emph{Complexity and Approximation: Combinatorial Optimization
  Problems and Their Approximability Properties}.\hskip 1em plus 0.5em minus
  0.4em\relax Springer. Berlin, 1999.

\bibitem{ICC11}
S.~Sorour and S.~Valaee, ``Completion delay minimization for instantly
  decodable network coding with limited feedback,'' \emph{IEEE International
  Conference on Communications (ICC'11)}, June 2011.

\end{thebibliography}

%%%%%%%%%%%%%%%%%%%%%%%%%%%%%%%%%%%%%%%%%%%%%%%

%%%%%%%%%%%%%%%%%%%%%%%%%%%%%%%%%%%%%%%%%%%%%%%%%%%%%%%%%%%%%%%%%%%%%%%%%

%%%%%%%%%%%%%%%%%%%%%%%%%%%%%%%%%%%%%%%%%%%%%%%%%%%%%%%%%%

\end{document}